\documentclass{llncs}
\usepackage[english]{babel}
\usepackage[T1]{fontenc}
\usepackage{latexsym,amsmath}
\usepackage{amssymb}
\usepackage[dvips]{graphicx}

\newtheorem{thm}{Theorem}

\title{Should Static Search Trees  Ever Be Unbalanced?}

\author{ Prosenjit Bose \and Karim Dou\"ieb\thanks{Research partially supported by NSERC and MRI.}}
\institute{ School of Computer Science, Carleton University, Herzberg Building \\
1125 Colonel By Drive, Ottawa, Ontario, K1S 5B6 Canada \\
\email{\{jit,karim\}@cg.scs.carleton.ca }\\
http://cg.scs.carleton.ca 
}

\date{}

\begin{document}

\sloppy
\maketitle

\begin{abstract}
In this paper we study the question of whether or not a static search tree should ever be unbalanced. We present several methods to restructure an unbalanced k-ary search tree $T$ into a new tree $R$ that preserves many of the properties of $T$ while having a height of $\log_k n +1$ which is one unit off of the optimal height. More specifically, we show that it is possible to ensure that the depth of the elements in $R$ is no more than their depth in $T$ plus at most $\log_k \log_k n +2$. At the same time it is possible to guarantee that the average access time $P(R)$ in tree $R$ is no more than the average access time $P(T)$ in tree $T$ plus $O(\log_k P(T))$. This suggests that for most applications, a balanced tree is always a better option than an unbalanced one since the balanced tree has similar average access time and much better worst case access time.
\end{abstract}

\section{Introduction}
The dictionary problem is fundamental in computer science, it asks for a data structure that efficiently stores and retrieves data. Binary search trees are simple, powerful and commonly used dictionaries.
The problem of building static search trees has been intensively studied in the past decades. Depending on the performance required one can build a perfectly balanced search tree that guarantees an optimal worst-case search time or one can build a biased search tree matching the entropy bound thereby providing an optimal expected search time. The search tree that minimizes the expected search cost can be unbalanced thereby behaving badly in the worst-case. Thus one may prefer to build a search tree of bounded height, i.e., with a certain guarantee on the worst-case search time that also minimizes the expected search time. In this paper we address the issue of the increase in the expected search cost imposed by restricting the height of the constructed tree. 

Since a search tree $T$ minimizing the expected search cost may behave badly in worst-case, one may want to construct another tree $R$ on the same set of keys in such a way that the worst-case search time is improved but the expected search time does not differ too much from the initial tree. One way to achieve this is to guarantee that $R$ has bounded height and that the depth of a key in $R$ is not much more than its depth in $T$. This is known as the \emph{restructuring search tree} problem.  Moreover, the problem of designing such search trees is directly related to the design of good codes. Thus the results obtained in this paper on search trees also has straightforward applications in coding theory.

\subsubsection*{Preliminaries}

Consider the set $x_1,x_2,\ldots,x_n$ of keys contained in a search tree $T$. We are given $2n+1$ weights $p_1,p_2,\ldots,p_n$ and $q_0,q_1,\dots,q_n$ such that $\sum_{i=1}^n p_i+\sum_{i=0}^n q_i=1$. Here, $p_i$ is the probability to query the key $x_i$ (successful search) and $q_i$ is the probability to query a key lying between $x_i$ and $x_{i+1}$ (unsuccessful search), $q_0$ and $q_n$ are the probabilities to query a key that is less or greater, respectively, than any key contained in the tree. 

Static \emph{multiway search trees} (or $k$-ary trees) generalize most of the other static search tree structures. A successful search ends up in an internal node of a $k$-ary tree that contains the requested key. Each internal node of a $k$-ary tree contains at most $k-1$ keys and has between 1 and $k$ children. An unsuccessful search ends up in one of the $n+1$ leaves of the $k$-ary tree. A leaf in a $k$-ary tree does not contain any key. The \emph{weighted path length} of a $k$-ary tree $T$ (referred to as path length in the remainder of this paper), a measure of the average number of nodes traversed during a search, is defined as 
\begin{equation}
\label{PL}
P(T)=\sum_{i=1}^n p_i (d_T(x_i)+1) + \sum_{i=0}^n q_i d_T(x_{i-1},x_{i}),
\end{equation}
where $d_T(x_i)$ is the depth in terms of number of links from the root node to the internal node containing the key $x_i$, $d_T(x_{i-1},x_{i})$ is the depth of the leaf reached at the end of the unsuccessful search for a key lying between $x_{i-1}$ and $x_i$. In the context of binary search trees (when $k=2$) in the comparisons-based model, the path length corresponds to the average number of comparisons performed during a search. In the external memory model, the path length corresponds to the average number of I/Os performed during a search in the case where each node is stored as one disk  block. Note that this is the usual way to store a multiway search tree in external memory. 

\subsection{Related work}


\subsubsection*{Optimal search trees}
Knuth~\cite{optbst} showed that an optimal binary search tree can be built in $O(n^2)$ time using $O(n^2)$ space. Mehlhorn~\cite{mehlornbestbound} gave an $O(n)$ time algorithm to build a binary search tree that is near-optimal. Concerning the more general case of $k$-ary trees, Vaishnavi~\emph{et al.}~\cite{optbtrees} showed that an optimal $k$-ary tree can be built in $O(k n^3)$ time. Becker~\cite{newoptbtrees} gave an $O(k n^\alpha)$ time algorithm, with $\alpha= 2+\log_{k} 2$, to build an optimal $B$-tree (subclass of $k$-ary tree) that satisfies the original constraints fixed by Bayer and McCreight~\cite{btrees}. These constraints require that every leaf in the $B$-tree have the same depth and that every internal node contains between $k/2$ and $k$ keys except for the root node. In the remainder of this paper, we consider a more general model of $k$-ary tree. The only constraint is that an internal node contains at most $k-1$ keys. Recently Bose and Dou\"ieb~\cite{nearoptbtree} presented a method to build a $k$-ary tree in $O(n)$ time (independent of $k$) that gives the best upper bound on the path length of a $k$-ary tree and produces a near-optimal $k$-ary tree for any $k\geq 2$. 

The problem of building an optimal search tree when only unsuccessful searches occur, i.e., when $\sum_{i=1}^n p_i=0$, is known as the optimal \emph{alphabetic search tree} problem. Hu and Tucker~\cite{n2optalphatrees} developed an $O(n^2)$ time and $O(n)$ space algorithm for constructing an optimal alphabetic binary search tree. This was improved by two other algorithms, the first one was by Knuth~\cite{knuthart} and the second by Garsia and Wachs~\cite{Garsiaoptalphatrees}. Both algorithms use $O(n \log n)$ time and $O(n)$ space. 

\subsubsection*{Optimal search trees with restricted height}
The problem of building an optimal binary search tree with restricted maximal height has been addressed by Garey~\cite{Garey}. The best algorithms solving this problem have been independently developed by Wessner~\cite{Wessner} and Itai~\cite{Itai}. They both produce the optimal binary search tree, with  $h$ as the height restriction, in $O(h n^2)$ time. For the problem of building an optimal alphabetic binary search tree with restricted maximal height $h$, Larmore and Przytycka~\cite{larmoreAlpha} presented a $O(hn \log n)$ time algorithm. 

\subsubsection*{Restructuring search trees}
The problem of restructuring a search tree $T$ consists of building another tree $R$, on the same set of keys, with restricted height such that the path length of $R$ is as close as possible to the path length of $T$. The \emph{drop} of a node $x$ is defined as $\Delta(x)=d_R(x)-d_T(x)$. This problem was initially posed by Bose. Evans and Kirkpatrick~\cite{restructuringordered} developed a technique to restructure a binary search tree $T$ into a tree $R$ of height $\lceil \log n \rceil +1$ such that $\Delta(x)\leq \log \log n$ for every node $x$ in $T$. They also showed that restructuring an alphabetic binary search tree can be done with the guarantee that $\Delta(x)\leq 2$ for every node $x$. Their work mainly focused on understanding the tradeoff between the height restriction of the restructured tree and the worst-case drop realized by a node. Gagie~\cite{Gagie} gave an alternate way to restructure a binary search tree into a tree of height $\log n +1$ that guarantees a slightly larger worst-case drop but aims at reducing the total drop as opposed to the worst case individual drop. He provided an algorithm where the path length of the restructured tree $R$ satisfies the following $P(R)\leq P(T)+(1+\epsilon)\log(P(T)+1)+\log((1/\epsilon)+1)+2$ with $1<\epsilon \leq 2$.

\subsection{Our results}
We present several methods to restructure a binary search tree that improves the previous best upper bounds on both the local drop of an individual node as well as the total drop of all nodes. The methods and the proofs are all based on a simple but general technique. We show that our method generalizes and are the first to study how to restructure multiway search trees (previous work only considers binary search trees).  Our results are then used to prove new tighter upper bounds on the path length of optimal height-restricted multiway search trees.

In Section~\ref{worstcasesection}, we present new tree restructuring methods that focus on reducing the worst-case drop of any given key. We first focus our attention on restructuring a given alphabetic $k$-ary search tree into another one of height $\log_k n +1$ such that at least a quarter of the leaves do not drop at all, the maximum drop realized by all but one of the leaves is at most $1$ and
exactly one leaf drops at most $2$ levels. Second, we present a restructuring method for the general case of $k$-ary search trees that builds another k-ary tree on the same keys with a guaranted worst-case drop of at most $\log_k \log_k n$. In fact, this method potentially gives a better bound since it takes into consideration the balance of the initial tree. The more unbalanced the initial tree, the better the guarantee on the drop. For example, if the initial tree is a path, then this method guarantees that the worst-case drop is at most 1.

In Section~\ref{depthpropdrop}, we develop a method focused on the relative drop. By this, we mean that in the worst case, the amount that a node will drop is proportional to its depth in the original tree as opposed to being proportional to the number of nodes in the tree. For a given node $x_i$, the maximum drop is at most $\log_{k} (d_T(x_i)+1) + (1+\epsilon) \log_{k} \log (d_T(x_i)+2)+\log_k \frac{1+\epsilon}{\epsilon}+1$. As a consequence of this, the path length of the restructured tree is close to the path length of the initial tree but the restructured tree has height at most $\log_k n +1$. In Section~\ref{hybrid} we combine the worst-case and relative drop approaches to obtain a hybrid method that guarantees simultaneously the best upper bounds in term of relative and worst-case drop plus a small constant. 

Finally we show in Section~\ref{upperheight} how the results on relative node drop can be used to obtain tighter upper bounds on the path length of optimal height-restricted multiway search trees.

\section{Restructuring multiway search trees}
Restructuring a search tree $T$ consists of building a new tree $R$, on the same set of keys, such that $R$ satisfies a precise constraint on its height. The problem is to determine how the tree $R$ differs from $T$ and how it is efficiently constructed. The main idea of our approach, similar to \cite{Gagie}, is to define a weight distribution on the keys based on their depth in the initial tree $T$. The weights of the keys are defined differently depending on what kind of guarantee on the drop we want to achieve. We distinguish between two types of guarantees on the drop: \emph{local} or \emph{global}. A local guarantee specifies the maximum drop realized by any node. A global guarantee specifies the maximum increase of the path length. Given these newly defined weights, we  build a near-optimal search tree using a technique described in the next section. 
 
\subsection{Method to construct near-optimal multiway search tees}
\label{nearlyopthere}
We describe a technique to build near-optimal multiway search trees, developed by Bose and Dou\"ieb \cite{nearoptbtree} and initially inspired from Mehlhorn's technique~\cite{mehlornbestbound} when access probabilities are known. This technique guarantees the best theoretical upper bound on the path length of optimal multiway search trees. 
Note that any other technique to build search trees can be used for the purpose of restructuring trees but we use~\cite{nearoptbtree} because it guarantees the best properties. 

Let $p_1,p_2,\ldots,p_n$ be the access probabilities of the internal nodes and $q_0,q_1,\dots,q_n$ be the access probabilities of the leaves. Let $T'$ be the tree built with the method \cite{nearoptbtree}. The following two lemmas characterize the depth of the elements in $T'$, we distinguish the cases where $T'$ has a branching number equal to $2$ or when it is greater. We define the value $m=\max\{n-3P,P\}-1\geq \frac{n}{4}-1$ where $P$ is the number of increasing or decreasing sequences in the access probability distribution on the ordered leaves. The value $q_{rank[i]}$ is the $i$th smallest access probability among the leaves except for the extremal ones (i.e. we exclude $(-\infty,x_1)$ and $(x_n,\infty)$ from consideration).


\begin{lemma}
\label{depthUB-Btrees}
The depth of the elements in $T'$ satisfy the following
\vspace{-0.2cm}
\begin{eqnarray*}
&&d_{T'}(x_i)\leq \lfloor \log_{k} \frac{1}{p_i+q_{min}}\rfloor \qquad \quad  {\bf for} \quad i=1,\ldots,n \, , \\
&&d_{T'}(x_{i-1},x_i)\leq \lfloor \log_{k} \frac{2}{q_i}\rfloor +1  \qquad \, {\bf for} \quad i=0,\ldots,n.
\end{eqnarray*}
\end{lemma}

The following lemma is not explicitly described in \cite{nearoptbtree}, additional details will appear in the journal version of this paper.

\begin{lemma}
\label{depthUB-BST}
In the case where $k=2$, the depth of the elements in $T'$ satisfy the following 
\vspace{-0.2cm}
\begin{eqnarray*}
&&d_{T'}(x_i)\leq \lfloor \log_{2} \frac{1}{p_i+q_{min}}\rfloor \qquad \, {\bf for} \quad i=1,\ldots,n \, , \\
&&d_{T'}(x_{i-1},x_i)\leq \lfloor \log_{2} \frac{1}{q_i}\rfloor +2  \quad\, \, {\bf for} \quad {\rm one \,\, leaf } \,(x_{i-1},x_i), \\
&&d_{T'}(x_{j-1},x_j)\leq \lfloor \log_{2} \frac{1}{q_j}\rfloor +1  \quad  \, {\bf for} \quad  {\rm all \,\, leafs} \,\,  (x_{j-1},x_j) \neq (x_{i-1},x_i), \\
&&d_{T'}(x_{j-1},x_j)\leq \lfloor \log_{2} \frac{1}{q_j}\rfloor   \quad \qquad {\bf for }\, \,{\rm at \,\, least }\, m+2 \, \,{\rm leaves} (x_{j-1},x_j).
\end{eqnarray*}
\end{lemma}

\begin{thm}
\label{UB-WADS}
The path length of the tree ${T'}$ is at most
\vspace{-0.2cm}
$$UB(k)=\frac{H}{\log_2 k}+1+\sum_{i=0}^n q_i - q_0 -q_n- \sum_{i=0}^m q_{rank[i]}, $$ 
where $H=\sum_{i=1}^n p_i \log_2 (1/p_i) + \sum_{i=0}^n q_i \log_2 (1/q_i)$ is the entropy of the probability distribution.  In the case $k=2$ the path length of ${T'}$ is at most
\vspace{-0.2cm}
$$
UB(2)=H+ 1- q_0 -q_n+q_{max}- \sum_{i=0}^{m'} pq_{rank[i]},
$$
where the value $m'=\max\{2n-3P,P\}-1\geq \frac{n}{2}-1$, $pq_{rank[i]}$ is the $i$th smallest access probability among every key and every leaf (except the extremal leaves) and $q_{max}$ is the greatest leaf probability including external leaves.
\end{thm}

\subsection{Worst case drop}
\label{worstcasesection}
In this section we consider the problem of minimizing the maximum drop independently realized by each node.
\subsubsection{Alphabetical tree}
~\\An alphabetic search tree is a tree where only unsuccessful searches occur, i.e., when $\sum_{i=1}^n p_i=0$. In order to restructure an alphabetic tree $T$, we first define a weight for each leaf in $T$ based on its depth in $T$. Namely the weight of a leaf node $(x_{i-1},x_i)$ is defined as
\vspace{-0.2cm}
$$
w(x_{i-1},x_i)=\max \left( \frac{1}{k^{d_T(x_{i-1},x_i)}},\frac{1}{(k-1)n} \right).
$$
Let $W=\sum_{i=0}^{n}w(x_{i-1},x_i)$ which is always strictly smaller than $1+\frac{n}{(k-1)n}=\frac{k}{(k-1)}$ by Kraft's inequality~\cite{Kraft}. These weights are used to define the access probabilities of each leaf. The access probability of a leaf $(x_{i-1},x_i)$ is defined as $q_i=w(x_{i-1},x_i)/W$ and the access probability of an internal node $x_i$ as $p_i=0$. These probabilities are then used as input to the algorithm described in Section~\ref{nearlyopthere} to build a near-optimal binary search tree giving the restructured tree $R$ on the same keys. 

\begin{thm}
\label{dropdepthaddalpha}
An alphabetic multiway tree $T$ can be restructured into a tree $R$ such that the height of $R$ is at most $\log_{k} n +1$ and the maximum drop of a leaf is at most $1$ if $k> 2$. When $k=2$ a drop of $2$ is realized by only one leaf, the drop of any other is at most $1$. In general, at least $m\geq\frac{n}{4}+2$ leafs do not drop.
\end{thm}
\begin{proof}
By Lemma~\ref{depthUB-Btrees}, 
the greatest depth reached by an internal node is $\lfloor \log_{k} \frac{1}{q_{min}}\rfloor< \log_{k} \frac{k(k-1)n}{(k-1)} =\log_{k} n +1$. As a consequence the greatest depth of a leaf is at most $\log_{k} n +1$, which corresponds to the maximum height of the restructured tree.

The depth of a leaf $(x_{i-1},x_i)$ in the restructured tree $R$ is at most $$\lfloor \log_{k} \frac{2}{q_i}\rfloor +1< \lfloor \log_{k} \frac{2k^{d_T(x_{i-1},x_i)+1}}{k-1}\rfloor +1= \lfloor \log_{k} \frac{2}{k-1}\rfloor + d_T(x_{i-1},x_i)+2.$$ Thus for $k>2$, the depth of a leaf $(x_{i-1},x_i)$ is at most  $d_T(x_{i-1},x_i)+1$ which implies a maximum leaf drop of 1. Using Lemma~\ref{depthUB-BST}, similar arguments verify the theorem in the case where $k=2$.
\qed
\end{proof}

So this simple method generalizes to k-ary alphabetic search trees the result of Evans and Kirkpatrick~\cite{restructuringordered}. It also gives a more precise guarantee on the maximal drop of a leaf in the binary alphabetic search tree case, since we guarantee that only one leaf drops two levels, all other leaves drop 1 level with a quarter of the leaves not dropping at all. Note that for some binary search trees any restructuring method produces a drop of 2 (see \cite{restructuringordered}). 

\subsubsection{General k-ary search tree}
\label{worstdrop}
~\\
Here the weight of an internal node $x_i$ is defined as follows
$$
w(x_i)=\max \left( \frac{1}{k^{d_T(x_i)}},\frac{W'}{(k-1)n} \right),
$$
where $W'=\sum_{i=1}^{n} \frac{1}{k^{d_T(x_i)}} \leq (k-1) \log_{k} n$ by the generalization of Kraft's inequality's~\cite{lowerBoundsBST}.  Let $W=\sum_{i=1}^{n}w(x_i)< W' + \frac{W'}{(k-1)}=\frac{k}{(k-1)} W'\leq k \log_{k} n.$ These weights are used to construct a probability distribution on the nodes. The access probability of an internal node $x_i$ is $p_{i}=w(x_i)/W$ whereas the access probability of a leaf is null, i.e., $q_i=0$ for all leaves. These probabilities are used to build the restructured tree $R$ with the technique described in Section~\ref{nearlyopthere}. 

\begin{thm}
\label{dropdepthaddalpha}
A multiway search tree $T$ can be restructured into a tree $R$ such that the height of $R$ is at most $\log_{k} n+1$ and the maximum drop of a node is at most $\lfloor \log_{k} \frac{W'}{k-1}\rfloor\leq \log_{k} \log_{k} n$.
\end{thm}
\begin{proof}
By Lemma~\ref{depthUB-Btrees}, the depth of an internal node $x_i$ is at most $\lfloor \log_{k} \frac{1}{p_i}\rfloor= \lfloor \log_{k} \frac{W}{w(x_i)}\rfloor$. The greatest depth reached by an internal  node is \vspace{-0.2cm} $$\max_i  \log_{k} \frac{W}{w(x_i)}< \log_{k}\frac{\frac{kW'}{(k-1)}}{\frac{W'}{(k-1)n}}= \log_{k} n +1. \vspace{-0.15cm}$$ As a consequence the greatest depth of leaf is at most $ \log_{k} n +1$, which corresponds to the maximum height of the restructured tree.
The depth of an internal node $x_i$ in the restructured tree $R$ is at most $$\lfloor \log_{k} \frac{1}{p_i}\rfloor< \lfloor \log_{k} \frac{k}{k-1}W'\, k^{d_T(x_i)}\rfloor=d_T(x_i)+\lfloor \log_{k} \frac{W'}{k-1}\rfloor+1.\vspace{-0.1cm}$$ The maximum drop is $\lfloor \log_{k} \frac{W'}{k-1}\rfloor \leq \log_{k} \log_{k} n$ for both internal nodes and leaves since the drop of a leaf is the same as the drop of its parent (an internal node). 
\qed \end{proof}

This method generalizes to k-ary search trees the result of Evans and Kirkpatrick~\cite{restructuringordered}. For the binary search tree case, the worst-case drop guaranteed with this method is similar to the one given by Evans and Kirkpatrick. Indeed there are some instances for which our method produces a drop of $\log_k \log_k n$. But for most instances the guarantee is better since our method takes into consideration the balance of the initial tree. For example if the tree is a list than the worst-case drop is constant. The value $W'$ is the expression of the balance of the initial tree, $W'$ is $O(1)$ for a highly unbalanced tree and $\Omega(\log n)$ when the tree is unbalanced.

\subsection{Relative drop}
\label{depthpropdrop}
Generally a static unbalanced search tree is needed when frequently accessed elements have to be accessed much faster than the other elements. In this context, if we want to restructure an unbalanced tree in order to satisfy a precise constraint on its height, it is important that elements  located close to the root in the original tree remain close to the root in the restructured tree. To achieve this, we bound the maximum drop of an element with respect to its depth in the original tree. This optimization differs from the previous one as it aims to reduce the global instead of local drop. 

First we define the weight of an internal element $x_i$ as
$$
w(x_i)=\max \left( \frac{1}{D(x_i) \, (d_T(x_i)+1)\, \log^{1+\epsilon}(d_T(x_i)+2)},\frac{1+\epsilon}{\epsilon n (k-1)} \right),
$$
with $1<\epsilon \leq 2$ and $D(x_i)$ is the number of elements at depth $d_T(x_i)$ in the tree $T$, thus $D(x_i)\leq (k-1)k^{d_T(x_i)}$. Let $W=\sum_{i=1}^{n}w(x_i)$ which is strictly smaller than $\sum_{i=1}^{n}\frac{1}{i \log^{1+\epsilon}(i+1)}+\frac{(1+\epsilon)n}{\epsilon \,n (k-1)}<\frac{k(1+\epsilon)}{(k-1)\epsilon}$. These weights define a probability distribution on the nodes so that the access probability of an internal node $x_i$ is given by $p_i=w(x_i)/W$. We consider the leaves to have an access probability of zero, i.e., $q_i=0$ for all leaves. These probabilities are used to build the restructured tree $R$ with the technique described in Section~\ref{nearlyopthere}. 
\begin{thm}
\label{dropdepth1}
Define $f(y)=\log_k y + (1+\epsilon)\log_k \log (y+1) + \log_k \frac{1+\epsilon}{\epsilon}+1.$ A multiway search tree $T$ can be restructured into a tree $R$ of height $ \log_{k} n +1$ where the drop of an internal node $x_i$ is at most $f(d_T(x_i)+1) $ and the drop of a leaf $(x_{i-1},x_i)$ is at most $f(d_T(x_{i-1},x_i))-1.$
\end{thm}
\begin{proof}
According to Lemma~\ref{depthUB-Btrees}, the depth of a internal node is at most $\lfloor \log_{k} \frac{1}{p_i}\rfloor= \lfloor \log_{k} \frac{W}{w(x_i)}\rfloor$. The greatest depth that an internal node can reach is  $$\max_i \log_{k} \frac{W}{w(x_i)}< \log_{k} \left(\frac{k(1+\epsilon)}{(k-1)\epsilon}\frac{\epsilon \, n(k-1)}{(1+\epsilon)}\right)=  \log_{k} n +1.$$ As a consequence the greatest depth of a leaf is at most $ \log_{k} n+1$, which corresponds to the maximum height of the restructured tree.

The depth of an internal node $x_i$ in $R$ is at most 
\begin{eqnarray*}
\lfloor \log_{k} \frac{W}{w(x_i)}\rfloor &<& \lfloor \log_{k} \frac{k(1+\epsilon)}{(k-1)\epsilon}D(x_i) \, (d_T(x_i)+1)\, \log^{1+\epsilon}(d_T(x_i)+2)\rfloor \\
&\leq &d_T(x_i)+\log_{k} (d_T(x_i)+1) + (1+\epsilon) \log_{k} \log (d_T(x_i)+2)+\log_k \frac{1+\epsilon}{\epsilon}+1.
\end{eqnarray*}
The maximum depth of a leaf in $R$ is the same as the maximum depth of its parent node in $R$. Thus the depth of a leaf $(x_{i-1},x_i)$ is at most $$d_T(x_{i-1},x_i)-1 +\log_{k} (d_T(x_{i-1},x_i)) + (1+\epsilon) \log_{k} \log (d_T(x_{i-1},x_i)+1) +\log_k \frac{1+\epsilon}{\epsilon}+1 .$$
\qed \end{proof}

\begin{thm}
\label{dropdepthh}
Define $m=\lfloor \log_{k}  n \rfloor+1$. A search multiway tree $T$ can be restructured into a tree $R$ such that the height of $R$ is at most $h$ (with $h\geq m$) and the depth of an internal node $x_i$ satisfies 
\begin{eqnarray*}
d_R(x_i)&\leq& d_T(x_i)+ f(d_T(x_i)+1-h+m) \quad {\bf if} \quad h-m\leq d_T(x_i)<h,\\
                &\leq& d_T(x_i) \quad  {\bf otherwise}. 
\end{eqnarray*}
For a leaf $(x_{i-1},x_i)$, 
\begin{eqnarray*}
d_R(x_{i-1},x_i)&\leq&  d_T(x_{i-1},x_i)+f(d_T(x_{i-1},x_i)-h+m) \quad {\bf if} \quad h-m\leq d_T(x_{i-1},x_i)<h,\\
                &\leq& d_T(x_{i-1},x_i) \quad  {\bf otherwise}. 
\end{eqnarray*}

\end{thm}
\begin{proof}
Consider the subtrees of $T$ rooted at the elements at depth $h-m$. Apply the restructuring procedure described in the beginning of this section to each of those subtrees seen as independent trees. This restructuring does not affect the depth of elements at depth strictly smaller than $h-m$. According to Theorem~\ref{dropdepth1}, the maximal drop of the other internal nodes $x_i$ is proportional to the depth inside the subtree that contains them, i.e., $d_R(x_i)\leq f(d_T(x_i)+1-(h-m))$. The maximum drop of a leave $(x_{i-1},x_i)$ is at most the maximum drop of its parent node, i.e., $d_R(x_{i-1},x_i)\leq f(d_T(x_{i-1},x_i)-(h-m))$.
\qed \end{proof}

We show how to restructure a tree $T$ into a tree $R$ with nearly minimum height such that the increase of the path length is small. This new restructuring tree method slightly improves the result of Gagie~\cite{Gagie} and arguably simplifies the proof technique (knowledge about relative entropy is not required). Evans and Kirkpatrick~\cite{restructuringordered} guaranteed a worst-case drop of $\log\log n$. Since this does not take into consideration the original depth of the element in the tree, this could lead to a situation where the depth of the root in the restructured tree is $\log\log n$ times greater then its depth in the initial tree.

\subsection{Hybrid drop}
\label{hybrid}
The first method presented in Section~\ref{worstdrop} gives the best upper bound on the worst case drop which is $\log_k \log_k n$. The problem is that the restructured tree produced by this method can have a path length which is $\log_k \log_k n$ times larger than the path length of the original tree. The method introduced in Section~\ref{depthpropdrop} avoids this problem by guaranteeing a drop that is proportional to the depth of the elements in the original tree, but the guarantee on the worst-case drop is a bit worst than the previous method. Here we present a hybrid method for restructuring a $k$-ary search tree that guarantees simultaneously the best upper bounds in term of relative and worst-case drop plus a small constant. 

Let $d'$ be the value that satisfies $(d'+1)\log^{1+\epsilon}(d'+2)=\log_k n$ with $1<\epsilon \leq 2$. The weight of an internal node $x_i$ is defined as follows
\vspace{-0.2cm}
$$
w(x_i)=
\begin{cases}
\max\left( \frac{1}{D(x_i) \, (d_T(x_i)+1)\, \log^{1+\epsilon}(d_T(x_i)+2)}, \frac{1+2\epsilon}{\epsilon n (k-1)}\right) & \text{for  $(d_T(x_i)+1)\leq d'$,}\\
\max\left( \frac{1}{D(x_i) \log_k n} , \frac{1+2\epsilon}{\epsilon n (k-1)}\right) & \text{for  $d'< (d_T(x_i)+1)\leq \log_k n$,}\\
\frac{1+2\epsilon}{\epsilon n (k-1)} & \text{for  $ (d_T(x_i)+1)> \log_k n$.}
\end{cases}
$$
The total weight is 
\vspace{-0.2cm}
\begin{eqnarray*}
W&=&\sum_{i=1}^n w(x_i)\\
&\leq&  \sum_{j=0}^{d'}\frac{1}{ (j+1)\, \log^{1+\epsilon}(j+2)} + \sum_{j=d'+1}^{\log_k n} \frac{1}{ \log_k n} + \frac{n(1+2\epsilon)}{\epsilon n (k-1)}\\
&<&  \frac{1}{\epsilon}+1 + 1+ \frac{(1+2\epsilon)}{\epsilon (k-1)}\\
&<& \frac{k(1+2\epsilon)}{(k-1)\epsilon}.
\end{eqnarray*}

Those weights are used to build the restructured tree $R$ with the technique described in Section~\ref{nearlyopthere}. The access probability of an internal node $x_i$ is given by $p_i=w(x_i)/W$ whereas the access probability of a leaf is null, i.e., $q_i=0$ for all leaves.

\begin{thm}
A k-aray search tree $T$ can be restructured into a tree $R$ such that the height of $R$ is at most $ \log_k n +1$ and the drop of a key $x_i$ is at most 
\vspace{-0.2cm}
$$\min\{ \log_k \log_k n, \log_k (d_T(x_i)+1)+(1+\epsilon)\log_k \log  (d_T(x_i)+2)  \} + \log_k \frac{1+2\epsilon}{\epsilon}+1.$$
\end{thm}
\begin{proof}
By Lemma~\ref{depthUB-Btrees}, the depth of an internal node $x_i$ is at most $\lfloor \log_{k} \frac{1}{p_i}\rfloor= \lfloor \log_{k} \frac{W}{w(x_i)}\rfloor$. The largest depth reached by an internal node is 

$$\max_i \, \log_{k} \frac{W}{w(x_i)}< \log_{k}\frac{\frac{k(1+2\epsilon)}{(k-1)\epsilon}}{\frac{1+2\epsilon}{\epsilon n (k-1)} }= \log_{k} n +1.$$ 
As a consequence the largest depth of a leaf is at most $\log_{k} n +1$ which corresponds to the maximum height of the restructured tree.
Using the same type of argument than in the proof of Theorem~\ref{dropdepth1}, an internal node $x_i$ with $ (d_T(x_i)+1)\leq d'$ realizes a drop of at most 
\begin{eqnarray*}
&&\log_k (d_T(x_i)+1) + (1+\epsilon)\log_k \log (d_T(x_i)+2) + \log_k \frac{1+2\epsilon}{\epsilon}+1\\
\end{eqnarray*}
which is at most $\log_k \log_k n + \log_k \frac{1+2\epsilon}{\epsilon}+1$ by the definition of $d'$.
An internal node $x_i$ with $d'< (d_T(x_i)+1)\leq \log_k n$ realizes a drop of at most 
\begin{eqnarray*}
&&\log_k \log_k n + \log_k \frac{1+2\epsilon}{\epsilon}+1\\
\end{eqnarray*}
which is at most $\log_k (d_T(x_i)+1)+(1+\epsilon)\log_k \log  (d_T(x_i)+2) + \log_k \frac{1+2\epsilon}{\epsilon}+1$ by the definition of $d'$.

\qed \end{proof}

\section{Applications}
\label{upperheight}
Nice applications of the results provided in the Section~\ref{depthpropdrop} about the relative drop occurs in the context of building optimal height-restricted multiway search trees. We are interested in measuring the maximum increase of the path length imposed by a height restriction. We investigate the difference between the path length of the optimal multiway tree and the optimal multiway tree with a height restriction. We give the best upper bound on the path length of an optimal multiway tree with a height restriction. Note that to prove the bound we assume that the access probabilities to the nodes and leaves are given. 

\begin{thm}
Consider $T^*$ the optimal multiway tree built over the set of keys $x_1, \ldots,x_n$ and let $T^*_h$ define the optimal multiway tree build on the same set of keys and such that its height is no more than $h\geq \lfloor \log_{k} n \rfloor+1$. The following is always satisfied
\vspace{-0.1cm}
$$
P(T^*_h)\leq P(T^*)+f(\max\{1,P(T^*)-h+m\}),
$$
where $f(y)=\log_k y + (1+\epsilon)\log_k \log (y+1) + \log_k \frac{1+\epsilon}{\epsilon}+1$ and $m=\lfloor \log_{k}  n \rfloor+1$.
\end{thm}
\begin{proof}
Using the method described in Section~\ref{depthpropdrop} we can restructure $T^*$ into the tree $R_h$ which has a maximum height $h$. By definition we have $P(T^*_h)\leq P(R_h)$. Using Theorem~\ref{dropdepthh} and by Jensen's inequality~\cite{Jensen} we show
\vspace{-0.2cm}
\begin{eqnarray*}
P(R_h)&=&\sum_{i=1}^n p_i (d_{R_h}(x_i)+1) + \sum_{i=0}^n q_i d_{R_h}(x_{i-1},x_{i})\\
&\leq&\sum_{i=1}^n p_i (d_{T^*}(x_i)+1+f(\max\{1,d_{T^*}(x_i)+1-h+m\})) \\
& & + \sum_{i=0}^n q_i (d_{T^*}(x_{i-1},x_i)+f(\max\{1,d_{T^*}(x_{i-1},x_i)-h+m\}))\\
&\leq&P(T^*)+f(\max\{1,P(T^*)-(h-m)\}).
\end{eqnarray*}
\qed 
\vspace{-0.2cm}
\end{proof}

Among other things this theorem states that a height restricted optimal multiway tree has a path length that differs from the optimal path length $P(T^*)$ without the height restriction by roughly $2
\log_k P(T^*)$ (even if the height restriction is nearly maximum, i.e., $\log n +1$). This casts doubt on the necessity of using unbalanced search trees.

\begin{thm}
There exists a linear running time algorithm which builds a multiway search tree $R_h$ with a height smaller than $h\geq \lfloor \log_{k} n \rfloor+1$ and such that  
$$
P(R_h)\leq  UB(k) + f(\max\{1,UB(k)-h+m\})
$$ 
where $UB(k)$ is defined in Theorem~\ref{UB-WADS}.
\end{thm}
\begin{proof}
We use the technique described in Section~\ref{nearlyopthere} to build a near-optimal multiway search tree $T$ in $O(n)$ time. This guarantees that $P(T)\leq UB(k)$. Then we restructure $T$ into $R_h$ in $O(n)$ time using the technique developed in Section~\ref{depthpropdrop}. We can deduce from Theorem~\ref{dropdepthh} the correctness of the theorem.
\qed \end{proof}

\section*{Acknowledgment}
The authors wish to thank Travis Gagie, Pat Morin and Michiel Smid for fruitful discussions. The
first author wishes to especially thank Luc Devroye for his many key insights on this problem
and for encouraging him to never give up.

\bibliographystyle{abbrv}    
\bibliography{biblio}

\end{document}